\documentclass{article}
\usepackage{graphicx} 

\usepackage[english]{babel}
\usepackage[utf8]{inputenc}
\usepackage[T1]{fontenc}
\usepackage{lmodern}
\usepackage[a4paper, margin=1in]{geometry}

\usepackage{amsmath, bbm, amsthm, graphicx, etoolbox, subcaption, stmaryrd}
\usepackage{thmtools,thm-restate}
\usepackage{epsfig}

\usepackage{amsfonts}
\usepackage{mathrsfs}
\usepackage{mathtools}
\usepackage{bbm}
\usepackage{verbatim}
\usepackage{footnote}

\title{Locally computing edge orientations}

\author{Slobodan Mitrovi\'c \\ University of California, Davis\thanks{Supported by the Google Research Scholar and NSF Faculty Early Career Development Program No.~2340048. \texttt{smitrovic@ucdavis.edu}}
\and
Ronitt Rubinfeld \\ Massachusetts Institue of Technology\thanks{Supported by NSF awards CCF-2006664, DMS-2022448 and CCF-2310818. \texttt{ronitt@csail.mit.edu}}
\and
Mihir Singhal \\ University of California, Berkeley\thanks{\texttt{mihirs@berkeley.edu}}
}

\newcommand{\f}{\frac}

\newcommand{\Ex}{{\mathbb{E}}}

\newcommand{\mcA}{{\mathcal{A}}}
\newcommand{\Aa}{{\mathcal{A}}}
\newcommand{\cA}{{\mathcal{A}}}

\newcommand{\mc}{\mathcal}

\newcommand{\eps}{\epsilon}

\newcommand{\p}[1]{\left( #1 \right)}
\newcommand{\paf}[2]{\left( \frac{#1}{#2} \right)}

\newcommand{\todone}[1]{}

\newcommand{\x}{\text}
\newcommand{\D}{\Delta}
\newcommand{\e}{\eps}
\newcommand{\Ot}{{\tilde{O}}}
\newcommand{\rb}[1]{\left( #1 \right)}

\usepackage{graphicx}
\usepackage{tikz}
\usepackage{framed}
\usepackage[framemethod=tikz]{mdframed}
\usepackage{todonotes}
\usepackage{color}

\definecolor{darkgreen}{rgb}{0,0.5,0}
\usepackage{hyperref}
\hypersetup{
    unicode=false,          
    colorlinks=true,        
    linkcolor=blue,          
    citecolor=purple,        
    filecolor=magenta,      
    urlcolor=cyan           
}

\usepackage[capitalize, nameinlink]{cleveref}
\usepackage[section]{placeins}

\newtheorem{theorem}{Theorem}[section]

\newtheorem{prop}[theorem]{Proposition}
\newtheorem{claim}[theorem]{Claim}

\crefname{thm}{Theorem}{Theorems}
\Crefname{lemma}{Lemma}{Lemmas}
\Crefname{invariant}{Invariant}{Invariants}
\Crefname{claim}{Claim}{Claims}
\Crefname{observation}{Observation}{Observations}
\Crefname{algorithm}{Algorithm}{Algorithms}
\Crefname{algocfline}{Algorithm}{Algorithms}
\Crefname{figure}{Figure}{Figures}
\Crefname{challenge}{Challenge}{Challenges}
\Crefname{prop}{Proposition}{Propositions}

\usepackage[linesnumbered,boxed, vlined]{algorithm2e}
	\DontPrintSemicolon
  \SetKwInOut{Input}{Input}
  \SetKwInOut{Output}{Output}
  \SetKwProg{function}{Function}{:}{}

\makeatletter
\newcommand*{\algotitle}[2]{%
  \stepcounter{algocf}%
  \hypertarget{algocf.title.\theHalgocf}{}%
  \NR@gettitle{#1}%
  \label{#2}%
  \addtocounter{algocf}{-1}%
}
\makeatother
  
\usepackage{lineno}
\nolinenumbers 
\usepackage{caption}
\usepackage{framed}
\usepackage{enumerate}

\usepackage{tikzsymbols}
\usepackage{thmtools}
\usepackage{nicefrac}

\DeclareMathOperator{\poly}{poly}

\begin{document}
\maketitle

\begin{abstract}
We consider the question of orienting the edges in a graph $G$ such that every vertex has bounded out-degree.
For graphs of arboricity $\alpha$, there is an orientation in which every vertex has out-degree at most $\alpha$ and, moreover, the best possible maximum out-degree of an orientation is at least $\alpha - 1$. 
We are thus interested in algorithms that can achieve a maximum out-degree of close to $\alpha$. 
A widely studied approach for this problem in the distributed algorithms setting is a ``peeling algorithm'' that provides an orientation with maximum out-degree $\alpha(2+\epsilon)$ in a logarithmic number of iterations.

We consider this problem in the local computation algorithm (LCA) model, which quickly
answers queries of the form ``What is the orientation of edge $(u,v)$?'' by probing the input graph. 
When the peeling algorithm is executed in the LCA setting by applying standard techniques, e.g., the Parnas-Ron paradigm, it requires $\Omega(n)$ probes per query on an $n$-vertex graph. 
In the case where $G$ has unbounded degree, we show that
any LCA that orients its edges to yield maximum out-degree
$r$ must use $\Omega(\sqrt n/r)$ probes to $G$ per query in the worst case, 
even if $G$ is known to be a forest (that is, $\alpha=1$).
We also show several algorithms with sublinear probe complexity when $G$ has unbounded degree. 
When $G$ is a tree such that the maximum degree $\Delta$ of $G$ is bounded, we demonstrate an algorithm that uses $\Delta n^{1-\log_\Delta r + o(1)}$ probes to $G$ per query. To obtain this result, we develop an edge-coloring approach that ultimately yields a graph-shattering-like result. 
We also use this shattering-like approach to demonstrate an LCA which $4$-colors any tree using sublinear probes per query.
\end{abstract}

\section{Introduction}
Orienting graph edges while obeying certain constraints has applications in many computational settings. 
For instance, low-out-degree orientation has been studied in the dynamic~\cite{brodal1999dynamic,he2014orienting,berglin2020simple,solomon2020improved,christiansen2022dynamic}, distributed~\cite{ghaffari2017distributed,fischer2017deterministic,ghaffari2018derandomizing,harris2019distributed,su2020distributed,harris2021locality}, and massively parallel computation~\cite{ghaffari2019improved,biswas2022massively} settings.
The general problem of graph orientation is of significant interest as it serves as an important algorithmic tool for other computational problems.
In their celebrated result, when the input is given as a rooted tree where each edge is oriented toward its parent, Cole and Vishkin~\cite{cole1986deterministic} show how to $3$-color a tree in only $O(\log^* n)$ many distributed rounds.
The algorithm of \cite{solomon2020improved} employs low-out-degree graph orientation to obtain a dynamic algorithm for graph coloring, and the works \cite{neiman2015matching,he2014orienting,bernstein2016faster} apply results from dynamic edge orientation in designing algorithms for matching.
To dynamically maintain spanners, the work of \cite{bodwin2016fully} develops a method that also relies on graph orientation. The authors of \cite{even2014deterministic} design a local computation algorithm (LCA) for bounded-reachability orientations (a different class of orientations) to develop an efficient LCA for coloring. 
It is well-known that the low-out-degree orientation and densest subgraph are problems that are dual to each other~\cite{charikar2000greedy,bahmani2014efficient}.
Moreover, graph orientation has applications in small subgraph counting and listing~\cite{bera2019linear,biswas2022massively}.
Our work focuses on low-out-degree orientation in the context of LCAs.

\subsection{The problem of \texorpdfstring{$r$}{r}-orientation}

Given a graph $G$ with $n$ vertices, we consider the problem of orienting its edges such that every vertex has out-degree at most $r$, for some parameter $r$. 
We call such an orientation an \textit{$r$-orientation}.
Letting $\alpha$ denote the arboricity and $\rho$ the pseudo-arboricity of $G$, the best possible achievable value of $r$ is $r=\rho$, and it is also known that $\rho \le \alpha \le \rho + 1$~\cite{picard1982network,westermann1988efficient,blumenstock2020constructive}. 
However, we are also interested in approximation algorithms where we can achieve some value $r \ge \rho$. Since $\rho$ and $\alpha$ differ by at most one, we will mostly compare $r$ to $\alpha$ instead of $\rho$ since it will make some of our algorithms easier to describe.

A line of work~\cite{ghaffari2017distributed,fischer2017deterministic,ghaffari2018derandomizing,su2020distributed,harris2019distributed,harris2021locality} studied this problem in the distributed LOCAL model in which, under different conditions on $\eps \alpha$, they demonstrate algorithms that use $O(\poly \log n)$ rounds to achieve a value of $r=\alpha(1+\eps)$. In particular, Su and Vu~\cite{su2020distributed} provide such an orientation in $\tilde{O}(\log^2 n / \eps^2)$ rounds, while Harris, Su, and Vu~\cite{harris2021locality} improve the dependence on $1/\eps$ from quadratic to linear at the expense of an additional $\log n$ factor in the round complexity.

We consider this problem in the local computation algorithm (LCA) model (fully defined in \cref{sec:lca}), in which the algorithm must be able to orient any input edge such that many copies of the algorithm will, with no interaction between them except for a shared random string, produce a consistent orientation.
To the best of our knowledge, the low-out-degree orientation problem has not been previously studied
in the LCA setting.

A common method for obtaining LCAs from distributed algorithms is the Parnas-Ron paradigm \cite{parnas2007approximating}. 
As we will discuss further in \cref{sec:related}, in the regime we consider, the Parnas-Ron paradigm does not give any nontrivial sublinear-time algorithm.

\subsection{Our results}

As our first result, we show that an LCA that finds an $r$-orientation (even when $r$ depends on $n$) requires at least $\Omega(n^{1/2}/r)$ probes per query, even if the algorithm is randomized and the input graph is a forest (with $\alpha=1$).
\begin{theorem}[Rephrasing of \cref{prop:lower-bound}] \label{thm:lower-bound-arb}
    For any parameter $r$, any LCA randomized algorithm that yields an $r$-orientation with probability at least 0.9 must use at least $\Omega(n^{1/2}/r)$ probes per query in the worst case.
\end{theorem}
In fact, this lower bound holds against an LCA model with relatively strong queries: we allow algorithms to make adjacency-list, adjacency-matrix, and degree probes.

We then show upper bounds for the problem of $r$-orientation in arboricity-$\alpha$ graphs of unbounded degree, with different algorithms for different regimes of $r$. These upper bounds are also polynomial in $n$, though there is some separation between them and the lower bound of $\Omega(n^{1/2}/r)$.

\begin{theorem}[Rephrasing of \cref{prop:ub-1,prop:ub-2,prop:ub-3}]
\label{thm:upper-bound-arb}
    Suppose there is a parameter $r$, and an input graph $G$ of arboricity $\alpha$. Then, if $r \ge 10(\alpha^2 n)^{1/3}$, there is a randomized LCA that can $r$-orient $G$ with at most $\Ot(\max\{\alpha n/r^2, 1\})$ probes per query. Moreover, if $\alpha=1$ (i.e., $G$ is a forest), then there is also an LCA which can do so for any $r$, with at most $\Ot(n/r)$ probes per query.
\end{theorem}
For example, in the case of a forest ($\alpha=1$) and when $r=10n^{1/3}$, this result gives an LCA for $r$-orientation with probe complexity $\Ot(n^{1/3})$. In contrast, the lower bound given by \cref{thm:lower-bound-arb} is $\Omega(n^{1/6})$.

Finally, we consider the bounded-degree case, where $G$ has maximum degree $\D$, in the specific case where $\alpha=1$ (that is, $G$ is a forest). In this case, any orientation achieves $r \le \D$, so the meaningful case to consider is where $r < \D$. We show a sublinear algorithm in the case where $r = \D^{1-\Omega(1)}$:

\begin{theorem}[Rephrasing of \cref{thm:ub-bdd}] \label{thm:ub-bdd-reph}
    Let $r$ be a parameter and $G$ be an input forest with maximum degree $\D$. Then, there is a randomized LCA that $r$-orients $G$, with at most $\D n^{1-\log_\D r + o(1)}$ probes per query.
\end{theorem}

To prove this theorem, we use a concentration bound on the size of connected components of a random subgraph of a bounded-degree forest.
We also use this to derive an LCA for 4-coloring any tree (or forest) with bounded degree in sublinear time.
Specifically, we show the following:

\begin{theorem}
    Let $G$ be an input forest with maximum degree $\D$, where $\D$ is a constant. Then, there is a randomized LCA that 4-colors $G$ with query complexity $O(n^{1-\beta})$, where $\beta > 0$ is a constant depending on $\D$.
\end{theorem}

\subsubsection{Future directions:}
It remains open whether similar upper bounds for $r$-orientation can be obtained for other families of graphs. We are optimistic that ideas akin to those we develop for \cref{thm:ub-bdd-reph} might yield new upper bounds for minor-free graphs.
This is because minor-free graphs do not expand, which intuitively is required for the coloring approach we design to prove \cref{thm:ub-bdd-reph}.
We also believe that useful tools in tackling this question include LCAs for minor-free partitioning oracles~\cite{hassidim2009local,levi2015quasi,kumar2021random}.
It would be quite interesting to obtain upper bounds for bounded arboricity graphs that go beyond our result \cref{thm:upper-bound-arb} or to bring the lower bound of \cref{thm:lower-bound-arb} up to match the result of \cref{thm:upper-bound-arb} in the case of a forest. 
We expect that techniques we develop to prove \cref{thm:ub-bdd-reph} do not transfer to bounded arboricity graphs, as there are even expanders with arboricity $3$.

\subsection{Related work} \label{sec:related}
\subparagraph{Orientation and Parnas-Ron paradigm:} 

The Parnas-Ron paradigm~\cite{parnas2007approximating} is an important tool for converting distributed algorithms for use in sublinear-time models.
In particular, if there is a $T$-round distributed local algorithm $\cA$, then applying the Parnas-Ron paradigm is equivalent to collecting the $T$-hop neighborhood of a vertex $v$ and using that neighborhood to obtain the output of $\cA$ for $v$, thus getting an LCA with query complexity $O(\D^T)$.

There is a classical peeling algorithm that, given a graph of arboricity at most $\alpha$, finds a $((2 + \eps) \alpha)$-orientation in $O(\log n)$ parallel or LOCAL rounds, for a constant $\eps > 0$. In each step, this algorithm considers all the vertices with degree at most $(2+\eps) \alpha$. All the edges incident to these vertices are oriented outward, and those vertices are then removed from the graph. It is easy to show that this algorithm takes $O(\log n)$ such steps. By following the Parnas-Ron paradigm, in graphs with maximum degree $\Delta$, this peeling algorithm can be simulated with $\Delta^{O(\log n)}$ LCA probes, which yields a trivial upper-bound for any $\Delta > 1$ --- this upper-bound amounts to gathering the entire graph.
If one aims for a more relaxed approximation, e.g., a $\sqrt{\Delta}$-orientation, then the aforementioned peeling algorithm takes $\log_{\sqrt{\Delta}}{n}$ many steps. Nevertheless, even in this case the Parnas-Ron approach yields $\Delta^{\log_{\sqrt{\Delta}}{n}} = n^2$ LCA probe complexity.
Moreover, suppose one aims for any LCA probe complexity that is sublinear in $n$. In that case, the peeling algorithm provides only the guarantee of $\Delta$-approximate orientation, which is trivial to obtain by simply orienting every edge arbitrarily.

\subparagraph{Orientation and LLL:} Several works study the LCA and the closely related VOLUME complexity of the distributed Lovász local lemma (LLL)~\cite{brandt2016lower,fischer2017sublogarithmic,dorobisz2021local,brandt2021randomized}.
Despite significant progress in understanding the randomized LCA complexity of sinkless orientation and LLL, those techniques do not seem to apply to the problem we study.
In particular, the LLL applies in cases where a random orientation will succeed at most vertices (as is the case with sinkless orientation). However, for the problem of $r$-orientation, a random orientation is very unlikely to be close to correct.

\subparagraph{LCA sparsification:}
One recent development in the LCA model is a graph sparsification technique, which considers only a carefully selected subset of the neighbors of a vertex $v$ to update $v$'s state. 
This approach has led to new advances for problems such as approximate matching~\cite{ghaffari2019sparsifying,kapralov2020space}, maximal independent set~\cite{ghaffari2019sparsifying,ghaffari2022local}, set cover~\cite{grunau2020improved}, and graph coloring~\cite{chang2019complexity}. However, it is unclear how to apply this technique in the context of $r$-orientation.

\subparagraph{LCA shattering:}
Our approach for bounded-degree forests is reminiscent of the graph-shattering technique, which has been very influential in designing efficient LCA algorithms for a number of problems, including approximate matching~\cite{levi2015local,barenboim2016locality,ghaffari2019sparsifying}, maximal independent set~\cite{rubinfeld2011fast,alon2012space,ghaffari2016improved,barenboim2016locality,ghaffari2019sparsifying}, graph coloring~\cite{barenboim2016locality,chang2019complexity} and LLL algorithms~\cite{alon1991parallel,molloy1998further,fischer2017sublogarithmic,chang2019distributed}. The ideas in this line of work can be traced back to Beck's analysis of the algorithmic LLL~\cite{beck1991algorithmic}. The general idea in these results is to first execute an algorithm that finds only a partial solution, e.g., finds an independent set that is not necessarily maximal. For properly designed algorithms, it can be shown that such an approach ``shatters'' the graph into relatively small connected components of interest, e.g., after removing the found independent set and its neighbors from the graph, each connected component has only $\poly(\Delta)$ vertices. These $\poly(\Delta)$ vertices in a single connected component can then be processed with $\poly( \Delta)$ LCA probes with a simple graph traversal.

We remark that our approach departs from this general scheme in the sense that it does not design an algorithm that finds a partial solution first. 
Instead, it immediately shatters the input tree into more manageable components.

\subparagraph{Other related work:}
Some prior work has also studied the complexity of LCA algorithms in unbounded-degree graphs. For instance, \cite{even2014deterministic} develop an LCA algorithm for $O(\Delta^2 \log \Delta)$-coloring a graph with maximum degree $\Delta = o(\sqrt{n})$. Another work \cite{levi2015local} designs two algorithms for graphs that do not have constant degrees: the first approach is for maximal independent set, which has a quasi-polynomial probe complexity in the maximum degree $\Delta$, while the second approach is for approximating maximum matchings and uses $\poly(\Delta, \log n)$ probes per query. 
Yet another line of work \cite{parter2019spanner, arviv2023spanner} develops an LCA algorithm to construct an $O(k^2)$-spanner with $\Ot(n^{1+1/k})$ edges with probe complexity $n^{2/3-\Omega_k(1)} \poly(\Delta)$.

Several works study the complexity of LLL and related graph problems in LCA, VOLUME, LOCAL, and locally checkable labeling (LCL) models~\cite{goos2015non,brandt2016lower,hefetz2016polynomial,rosenbaum2020seeing,brandt2021randomized,grunau2022landscape,balliu2022distributed}.
We note that our lower-bound applies to the LCA setting, which, as we discuss in \cref{sec:preliminaries}, also applies to the VOLUME model.
The authors of \cite{brandt2021randomized} show that \textit{deterministically} coloring a tree with maximum degree $\poly(c)$ with $c \ge 2$ colors, for any fixed constant $c$, requires $\Theta(n)$ VOLUME probes. Another line of work on the hierarchies of VOLUME complexities~\cite{rosenbaum2020seeing} shows that there are graph problems whose randomized VOLUME complexity is $\Theta\rb{n^{c}}$, for any constant $c \in (0, 1)$. The celebrated result~\cite{parnas2007approximating} proves that to approximate a vertex cover within a constant factor multiplicative and $\eps n$ additive error it is necessary to use at least the order of the average degree many probes.

\section{Preliminaries}
\label{sec:preliminaries}

\subsection{LCA model} \label{sec:lca}
Throughout this paper, we use the LCA model of computation, which was originally introduced in \cite{rubinfeld2011fast,alon2012space}. The goal of this model, in short, is to orient any input edge $e$ quickly --- in particular, without necessarily computing the orientations of the rest of the graph.

Formally, in this model, the algorithm $\mc A$ receives, as a query, an edge $e$ of the graph $G$. It also has probe access to an adjacency matrix of $G$, an adjacency list of $G$, and the degrees of the vertices of $G$. Specifically, for any $i, u, v$, $\mc A$ can probe the following: whether $\{u, v\}$ is an edge of the graph, the $i$-th neighbor of $v$, or the degree of $v$. $\mc A$ is assumed to know the number of vertices and their IDs beforehand. $\mc A$ also has access to a source of shared randomness, which does not use up probes to access. Then, after making some probes to the input, $\mc A$ must return an orientation of $e$. Suppose that for each edge of $G$, one separate copy of $\mc A$ was run, with the same source of shared randomness across all copies (but no ability to communicate otherwise). 
Then, we wish that the resulting orientation should satisfy the desired property, e.g., low out-degree for every vertex, with a given (high) probability. We will be concerned with the \textit{probe complexity} of such an algorithm, i.e., the maximum number of probes it can possibly use on any edge. In this paper, we will not be concerned with the time or space complexity of our algorithms. 
In particular, we will not limit how long the shared random string may be.

Though we have defined the model to have access to adjacency matrix probes, all of the algorithms we describe do not use this kind of probe, instead only using adjacency list and degree probes. Our lower bounds, on the other hand, are still robust even against algorithms that use adjacency matrix probes.

The VOLUME model, introduced in \cite{rosenbaum2020seeing}, is another computation model which is similar to, but slightly weaker than, the LCA model. Since it is strictly weaker, this means that our lower bounds also carry over to the VOLUME model.

There has been significant interest in designing LCA algorithms for fundamental problems in computer science. Some examples include
locally (list)-decodable codes~\cite{BF90,L89,GLR+91,FF93,KT00,Yek10,
GRS00,AS03,STV01,GKZ08,IW97,KS09,BET10},
local decompression~\cite{MuthuSZ,SG,FV07,GN},
local reconstruction and filters for monotone and
Lipshitz functions~\cite{ACC+08,SS10,BhattacharyyaGJJRW12,JhaR13,
AwasthiJMR12,Lange2022ProperlyLM},  and
local reconstruction of graph properties
\cite{KalePS13}.
The study of LCAs has been very active in the
past few years, with recent results that include constructing
maximal independent sets 
\cite{rubinfeld2011fast,alon2012space,barenboim2016locality,LeviRY17,even2014deterministic,ghaffari2016improved,ghaffari2019sparsifying,ghaffari2022local}, coloring \cite{chang2019complexity}, approximate maximum matchings \cite{NO08, YYI09,
MRVX12,MansourV13,LeviRY17, ghaffari2019sparsifying, KapralovMNT19}, satisfying assignments for $k$-CNF \cite{rubinfeld2011fast,alon2012space}, local computation mechanism design \cite{HassidimMV16},
local decompression \cite{DLRR13}, local reconstruction of graph properties \cite{CampagnaGR13}, minor-free graph partitioning~\cite{hassidim2009local,levi2015quasi,kumar2021random}, and local generation of large random objects \cite{EvenLMR17,biswas2020local,biswas2022local}.

\subsection{Tools and notation} 

We state Yao's minimax principle, applied specifically to LCAs in the form that we will use.
\begin{theorem}[Yao's minimax principle] \label{thm:yao}
Suppose there exists a randomized LCA $\Aa$ with worst-case probe complexity $f(n)$ which solves a problem with probability $p$. Then, for any distribution $\mc D$ of inputs, there is a deterministic LCA $\Aa$ with probe complexity $f(n)$ which with probability $p$ solves the problem on an input drawn from $\mc D$.
\end{theorem}

Given a graph $G = (V, E)$, we use $\deg_G u$ to denote the degree of vertex $u \in V$. When $G$ is clear from the context, we omit the subscript $G$ and write $\deg u$ only.

When we use little-$o$ asymptotic notation, it is assumed to be with respect to the variable $n$, the number of vertices in the input graph. We will also frequently write inequalities that only hold when $n$ is sufficiently large, implicitly using the fact that our asymptotic statements are usually trivially true when $n$ is bounded. Moreover, we will often omit floor and ceiling signs for clarity.

We use the phrase ``with high probability'' to mean ``with probability at least $1-n^{-c}$,'' where $c > 1$ is a constant that is sufficiently large such that all union bounds which are used henceforth will still yield sufficiently small probabilities. This is a slight abuse of notation, since the threshold for $c$ will vary in each usage, but the meaning will be clear from context. We also use the phrase ``with very high probability'' to mean ``with probability at least $1-f(n)$,'' where $f(n)$ is smaller than $n^{-c}$ for any constant $c$, for sufficiently large $n$.


\section{Orientation in unbounded-degree graphs}

We first consider the question of $r$-orientation in general graphs of arboricity at most $\alpha$, with no further restriction on the graph. 

\subsection{Lower bound for orientation in unbounded-degree graphs}

First, we demonstrate a lower bound that shows that even in the $\alpha=1$ (i.e., forest) case, we need $\Omega(n^{1/2}/r)$ probes to $r$-orient a graph. Essentially, the proof constructs a random tree in which it is hard to find a particular star hidden inside the tree, so the edges of the star must be oriented randomly. This lower bound is particularly powerful since it works even for an LCA that can probe both the adjacency lists and the adjacency matrix (as well as the degrees).

\begin{theorem}\label{prop:lower-bound}
Suppose that there is an LCA which, on an input forest $G$ with $n$ vertices, $r$-orients edges of $G$ with probability greater than 0.9. Then, the LCA must use $\Omega(n^{1/2}/r)$ probes per query in the worst case.
\end{theorem}
\begin{proof}
We assume that $r \le 0.001n^{1/2}$ since otherwise the statement is obvious. Suppose for contradiction that an LCA $\mcA$ exists which can perform the required orientation, using less than $0.001n^{1/2}/r$ probes per query.

We describe the construction of a random graph $G$, which we will use as input to $\mcA$. Define the parameters $s = 24r$ and $t = n^{1/2}/4s$. Note that $t > 10$. We construct $G$ as follows. 

First, we have a set $A$ of $|A|=st$ vertices. Each $a \in A$ is associated with a set $S_a$ of $|S_a|=st$ vertices, whose vertices are all connected to $a$ (forming a star). Also, there is a larger set $B$ of $|B| = s^2 t - 2s$ vertices; each $b \in B$ has a corresponding set $S_b$ with $|S_b| = t$ vertices, whose vertices are all connected to $b$. These constitute all the vertices in the graph. The total number of vertices in the graph is then $(st+1)|A| + (t+1)|B| < 4s^2t^2 < n$. So far, all the edges we have described are deterministic. 

We describe three different types of edges; for ease of reference, we assign each type as a color. Let the deterministic edges we have drawn so far be black. We refer to all other edges as colored edges.

\begin{figure}
    \begin{center}
        \includegraphics[width=0.9\linewidth]{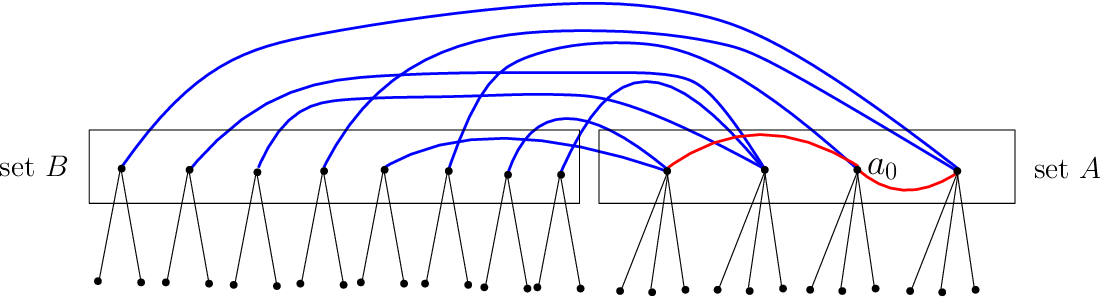}
    \end{center}
    \caption{Graph $G$ used to demonstrate lower bound}
    \label{fig:lbd}
\end{figure}

\begin{description}
    \item[Center vertex:] Pick a vertex $a_0 \in A$ uniformly at random; we refer to this vertex as the 
    \textit{center}. 
    \item[Red edges: ] Pick a set $A' \subseteq A$ of size $s$ not containing $a_0$ uniformly at random, and for each $a \in A'$, draw a red edge between $a$ and $a_0$; so, the red edges form a star centered at $a_0$. These red edges will form the ``hidden star'' which it will be difficult for the LCA to find.
    \item[Blue edges:] For each vertex $a \in A$, suppose that $a$ currently has degree $k$ ($k$ must equal one of $st$, $st+1$, or $st+s$, since $a$ has $st$ black edges and $0$, $1$, or $s$ red edges). Then, connect $a$ to $st + s - k$ vertices in $B$, so that each vertex in $B$ is connected to a single vertex in $A$. This matching is picked uniformly at random. The size of $B$ is exactly the sum of $st + s - k$ over all vertices $a$, so every vertex in $B$ can indeed have exactly one blue edge. The purpose of these blue edges is primarily to ensure that the vertices of $G$ have fixed degree, so that degree probes will give no extra information.
    \item[Permuted neighbors:] The adjacency list for each vertex is also generated as a random permutation of its neighbors.
\end{description}

This completes the description of the graph $G$. An example rendition of this graph is shown in \cref{fig:lbd}. The red edges form a star centered at $a_0$; in order to get an $r$-orientation, many of these edges will have to be oriented away from $a_0$.

We now outline the remainder of the proof. We will show that the LCA $\Aa$ cannot keep the out-degree of $a_0$ below $r$ with probability 0.9; in particular, we show that with probability over 0.1 it must orient at least $r$ of the red edges toward $a_0$. To prove this, we consider what $\Aa$ does when it receives a red edge as a query. We will show that, since the black edges constitute most of $G$, $\Aa$ will (with some probability) never be able to probe any colored edge at all and thus must return a deterministic orientation that is independent of which of the endpoints of the queried edge is $a_0$. 
It will follow that in expectation, many red edges must be oriented away from $a_0$, so the expected out-degree of $a_0$ will be large, which will complete the proof.

By Yao's minimax principle (stated in \cref{thm:yao}), we may assume that $\mc A$ is deterministic; it suffices to show that it is not possible that with probability at least $0.9$ (over the randomness of the graph $G$), the out-degree of $a_0$ is at most $r$.
Every vertex in $G$ has a fixed degree, even though some of the edges in $G$ are chosen randomly, so we may further assume that $\mc A$ never makes a degree probe.

We show that for each red edge, upon probing it with $\mc A$, there is at least a $1/3$ chance it orients toward the center $a_0$.

Consider the action of $\mc A$ upon receiving a probe of the red edge $(u, v)$, where $u, v \in A$, and the distribution of the graph $G$ is conditioned on $(u, v)$ being an edge of $G$ (so either $u$ or $v$ is $a_0$, each with probability $1/2$, and the other is in $A'$). We claim that with probability at least $1/6$, $\mc A$ never probes any colored edge (we assume that it does not probe $(u, v)$ in the adjacency matrix since it knows this edge is in $G$). As a reminder, we are assuming that $\mc A$ uses less than $0.001n^{1/2}/r$ probes per probe.

Consider the $i$-th probe that $\mc A$ makes, assuming that it has not probed any colored edges yet in the previous $i-1$ probes. If the first $i-1$ probes have not revealed any colored edges, then their results are deterministic and fixed, so the $i$-th probe is also deterministic (since we have assumed that $\mc A$ is deterministic). We claim that the probability that this fixed probe yields a colored edge is at most $1/t$. Since the probe is deterministic, it must be an adjacency list probe at a fixed index for a fixed vertex, or an adjacency matrix probe for two fixed vertices (every vertex has a fixed degree in $G$, so a degree probe does not reveal any information). We do casework on each of the possible probe types:
\begin{itemize}
\item Adjacency list probe for any vertex in $A$ (including $u$ or $v$): Each vertex in $A$ has $s$ colored edges and $st$ black edges. Since the adjacency list is permuted randomly, the probability that the probe selects a colored edge is $s/(st+s) < 1/t$.
\item Adjacency list probe for any vertex in $B$: Each vertex in $B$ has $1$ colored edge and $t$ black edges, so again, the probe selects a colored edge with probability $1/(t+1) < 1/t$.
\item Adjacency list probe for a vertex in $S_a$ or $S_b$ for some $a \in A$ or $b \in B$: These vertices have no colored edges, so the probe cannot return a colored edge.
\item Adjacency matrix probe between either $u$ or $v$ and another vertex in $A$: either $u$ or $v$ is $a_0$ with equal probability, and the probability that $a_0$ is connected to any other given vertex in $A$ (conditioned on it already being connected to the other one of $u$ and $v$) is $(s-1)/(st-2)$, so the probability that this edge returns a colored (red) edge is $(s-1)/2(st-2) \le s / 2st < 1/t$.
\item Adjacency matrix probe between any vertex $a \in A$ and any vertex $b \in B$: for any fixed $a \in A$, conditioned on the selection of the red edges, the neighbors of $a$ in $B$ are a uniformly random subset of $B$ of some fixed size which is at most $s$. Therefore, the probability that this probe returns an edge is at most $s / |B| = s / (s^2t - 2s) < 2s / s^2 t < 1/t$. 
\item Any other adjacency matrix probe: these cannot reveal any colored edge.
\end{itemize}

So we have shown that the probability that the first $i-1$ probes do not reveal any colored edges but the $i$-th probe does is at most $1/t$. It follows that the probability that any probe reveals a colored edge is at most $(0.001n^{1/2}/r)(1/t) < 1/3$.

Therefore, with probability $2/3$, when orienting edge $(u, v)$, conditioned on $G$ containing that edge, $\mc A$ does not probe any colored edges, and thus returns a deterministic orientation. Since $u$ and $v$ are each $a_0$ with (conditional) probability $1/2$, this means that $\mc A$ orients this edge away from $a_0$ with probability at least $2/3-1/2=1/6$. 

Now, the expected out-degree of $a_0$ is at least
\begin{align*}
&\quad \sum_{u, v \in A} \Pr[(u, v) \in G] \cdot \Pr[\x{$(u, v)$ oriented away from $a_0$} \mid (u, v) \in G] \\
&\ge \f 16 \sum_{u, v \in A} \Pr[(u, v) \in G] = \f 16 \Ex[\deg a_0] = \f s6 = 4r.
\end{align*}
We have shown that the expected out-degree of $a_0$ is at least $4r$, but what we actually want to show is that it cannot be at most $r$ with probability 0.9. Indeed, the out-degree of $a_0$ is bounded above by $s=24r$, so if it were most $r$ with probability at least $0.9$, then the expected out-degree of $a_0$ would be at most $0.9r + 0.1 \cdot 24r < 4r$, which would be a contradiction. 
This concludes the proof.
\end{proof}

\subsection{Algorithms for orientation in unbounded-degree graphs}
With the $\Omega(n^{1/2}/r)$ lower bound in mind that follows from \cref{prop:lower-bound}, we demonstrate some algorithms that achieve comparable upper bounds. Specifically, for $r$-orientation in forests, we can achieve LCA complexities of $O(1)$ when $r \ge n^{1/2}$, $\Ot(n/r^2)$ when $10n^{1/3} \le r \le n^{1/2}$, and $\Ot(n/r)$ otherwise.

\subsubsection{\texorpdfstring{$O(1)$ algorithm for large $r$}{O(1) algorithm for large r}}

First, we show the case where $r \ge n^{1/2}$. In this case, the algorithm is very simple: just orient edges toward the higher degree vertex. This also works with graphs of higher arboricity, with a tradeoff in $r$.
\begin{prop} \label{prop:ub-1}
There exists an LCA that can $(2\alpha n)^{1/2}$-orient an $n$-vertex $\alpha$-arboricity graph with probe complexity $O(1)$.
\end{prop}
\begin{proof}
Given edge $\{u, v\}$ with $\deg v \ge \deg u$, orient the edge from $v$ to $u$. (If the degrees are equal, then return an arbitrary orientation.) Suppose for the sake of contradiction that there is a vertex $v$ with out-degree at least $\sqrt{2 n \alpha}$. Then, it has at least $\sqrt{2 n \alpha}$ neighbors each with a degree at least $\sqrt{2 n \alpha}$, so the total degree of the graph is at least $2n\alpha$. But a graph with arboricity $\alpha$ has at most $\alpha(n-1)$ edges, a contradiction.
\end{proof}

\subsubsection{\texorpdfstring{$\Ot(\alpha n/r^2)$ algorithm for medium $r$}{\tilde O(alpha n/r\^{}2) algorithm for medium r}}

The following shows that we can use a similar idea for $r \ge 10n^{1/3}$.
We can orient all edges toward vertices that have a very large degree since there are very few of these vertices. With vertices of small degrees (less than $r$), the condition is trivially satisfied, so we can also orient all edges away from these vertices. This leaves the medium-degree vertices, of which there also cannot be too many.
We then essentially consider the subgraph induced by these medium vertices and repeat the previous idea, orienting toward the high-degree vertices. Again, this generalizes to arboricity $\alpha$.

\begin{prop} \label{prop:ub-2}
Suppose a graph $G$ with $n$ vertices and arboricity $\alpha$ is given as input.
Then, for all $10(\alpha^2 n)^{1/3} \le r \le (\alpha n)^{1/2}$, there is an LCA which, with high probability, $r$-orients the edges in $G$ with $\Ot(\alpha n/r^2)$ probes per query.
\end{prop}
\begin{proof}
We describe the algorithm. Let $s = r/10$. Call a vertex \textit{small} if its degree is at most $r$, \textit{large} if its degree is at least $\alpha n/s$, and \textit{medium} otherwise.

Suppose the algorithm receives a query of $e$ with endpoints $u, v$ (assume arbitrarily that $u$ has the smaller ID). Then it proceeds as shown in \cref{alg:medium-r}, splitting into three cases based on the degrees of its endpoints. The probe complexity in Cases 1 and 2 is $O(1)$, and in Case 3 is $\Ot(d/s) = \Ot(\alpha n/r^2)$, as desired. It remains to check correctness of the algorithm.

\begin{algorithm} \label{alg:medium-r}
    \Input{Graph $G$ with $n$ vertices and arboricity $\alpha$, with probe access \\
        Query edge $e=(u, v)$, where $u$ has lower ID than $v$
    }
    
    Perform degree probes on $u$ and $v$.

    \uIf{$u$ or $v$ is small} {
        \tcp{Case 1}
        Orient $e$ away from that vertex. (Pick arbitrarily if both are small.)
    }
    \uElseIf{$u$ or $v$ is large} {
        \tcp{Case 2}
        Orient $e$ toward that vertex. (Pick arbitrarily if both are large.)
    }
    \Else{
        \tcp{Case 3 ($u$ and $v$ are both medium)}
        
        Let $d = \deg u$. \tcp{Note that $r \le d \le \alpha n/s$}
        
        Sample $(1000d\log n)/s$ neighbors of $u$ independently and uniformly at random (using a fixed part of the shared randomness, depending on $e$).
        
        Perform a degree probe on each sampled neighbor.
        
        \uIf{at most a $2s/d$ proportion of the sampled neighbors are medium} {
            Orient $e$ away from $u$.
        }
        \Else {
            Orient $e$ toward $u$.
        }
    }

	\caption{
        LCA for $r$-orienting a graph when $r \ge 10(\alpha^2 n)^{1/3}$
    }
\end{algorithm}

\subparagraph{Correctness.} We wish to check that the out-degree of every vertex is at most $r$ with high probability. Small vertices obviously have out-degree at most $s \le r$, so we only need to check the medium and large vertices.

The total degree of $G$ is less than $2\alpha n$, so there are at most $2s$ large vertices. Thus, large vertices end up with an out-degree of $2s$ at most since edges from large to non-large vertices are oriented toward the large vertex.

It remains to check the medium vertices. Since the total degree of $G$ is $2\alpha n$, there are also at most $2\alpha n/r$ medium vertices. 

Now, consider a medium vertex $w$. No edges would have been oriented away from $w$ in Case 1, and since there at most $2s$ large vertices, there are at most $2s$ edges oriented away from $w$ in Case 2.

It remains to analyze the number of edges that are oriented away from $w$ in Case 3 (that is, the medium-to-medium edges). Note that by a Chernoff bound, with high probability, Case 3 orients a medium-to-medium edge $(u, v)$ away from $u$ if it has at most $s$ medium neighbors, and toward $u$ if it has at least $3s$ medium neighbors. We then show that there are at most $3s$ edges oriented away from $w$ in Case 3.

If $w$ has at most $3s$ medium neighbors, then this is obvious, so suppose $\deg w > 3s$. Consider the action of the algorithm when orienting an edge $e$ containing $w$. When we input $e=(u, v)$ into the algorithm, it is possible for either $u$ or $v$ to be $w$. If $u = w$, then it (with high probability) orients the edge toward $w$. Otherwise, if $v = w$, then (with high probability) $u$ must be a medium vertex with at least $s$ medium neighbors in order to orient $e$ toward $w$. Now, the subgraph of $G$ induced by the medium vertices also has arboricity $\alpha$, and thus has total degree at most $4\alpha^2 n/r$ by the bound on the number of medium vertices. Therefore there are at most $4\alpha^2 n/rs$ medium vertices $u$ with at least $s$ medium neighbors. Therefore, the number of edges oriented away from $w$ in Case 3 is at most $4\alpha^2 n/rs \le 3s$.

Thus, in total, every medium vertex has at most $5s < r$ edges oriented away from it, completing the proof.
\end{proof}

\subsubsection{\texorpdfstring{$\Ot(n/r)$ algorithm for all $r$, if $G$ is a forest}{\tilde O(n/r) algorithm for all r, if G is a forest}}

Finally, we consider the case of small $r$, in the case where $G$ is a forest. Here, we can still orient edges toward very large-degree vertices. Then, we randomly color all the remaining edges with $r$ colors and orient each color class separately (with maximum out-degree $1$ within each color class). We show that the color classes are small with high probability, bounding the LCA complexity. Note that this algorithm only applies to arboricity $\alpha=1$, i.e., the case of a forest.

\begin{prop} \label{prop:ub-3}
For all $r \le n^{1/2}$, there exists an LCA which, with high probability, $r$-orients edges in an $n$-vertex forest $G$ which uses $\Ot(n/r)$ probes per query.
\end{prop}
\begin{proof}
Color each edge randomly with one of $r/5$ colors (using the shared randomness). Say that a vertex is \textit{large} if its degree is at least $5n/r$. Also, define an edge to be large if one of its vertices is large. Then, we proceed as described in \cref{alg:large-r}.


\begin{algorithm} \label{alg:large-r}
    \Input{Forest $G$ with $n$ vertices, with probe access \\
        Query edge $e=(u, v)$, where $u$ has lower ID than $v$
    }
    
    Using the shared randomness, assign an independent and uniformly randomly chosen color from $\{1, \dots, r/5\}$ to every pair of vertices in $G$ (so that every edge computes the same coloring).
    
    Perform degree probes on $u$ and $v$.
    
    \uIf{$u$ or $v$ is large} {
        Orient $e$ toward that vertex. (Pick arbitrarily if both are large.)
    }
    \Else {
        Let $c$ be the color of $e$.
        
        Perform a depth-first search from each of the endpoints of $e$ to find the connected component of the edge of color $c$, but stop the search at any large vertex. 
        
        Orient $e$ toward the minimum-ID vertex in the connected component (which is a tree). \tcp{Note that together, the edges in the connected component of color $c$ form an orientation of that connected component with out-degree at most 1.}
    }
    
	\caption{
        LCA for $r$-orienting a forest for any $r$
    }
\end{algorithm}

Since $G$ has a total degree at most $2n$, there are at most $2r/5$ large vertices, so at most $2r/5$ large edges are oriented away from each vertex. Also, at most one non-large edge of each color is oriented away from each vertex. So, in total, each vertex has an out-degree less than $r$. 

Thus, this algorithm gives a valid orientation. It remains to see that it has low probe complexity. Let $u$ be one of the endpoints of $e$. It is enough to check that the depth-first search starting at $u$ searches at most $\Ot(n/r)$ edges with a very high probability (by symmetry, the same will hold for the other endpoint of $e$).

Indeed, let $H$ be the connected component of $G \setminus e$ containing $u$; consider it as a tree rooted at $u$. For a vertex $v \neq u$ in $H$, let $X_v$ be the indicator random variable of whether the edge from $v$ to its parent has color $c$. 
Then, for the search to have to check the edges from $v$ to its descendants, $X_v$ must equal $1$ (and $v$ must also not be large). The search also always checks the edges from $u$ to its descendants (as long as $u$ is not large); there are at most $O(n/r)$ such edges. Therefore, the total number of edges checked is bounded above by the following quantity:
\[O(n/r) + \sum_{\x{non-large } v \in H \setminus u} X_v \deg v. \]
Since $\deg v \le 5n/r$ for each non-large $v$, the above sum is a sum of random variables bounded in $[0, 5n/r]$. The $X_v$ are all independent Bernoulli random variables with probability $5/r$, and the total degree of all vertices is at most $2n$. The expectation of the sum is then $O(n/r)$. Then, by a Chernoff bound, the sum is $\Ot(n/r)$ with high probability, so we are done.
\end{proof}

\section{Orientation in bounded-degree forests}
\label{sec:bounded-orientation}
Finally, we present an algorithm that $r$-orients a forest that has maximum degree $\D$, using at most $\D n^{1 - \log_\D r + o(1)}$ probes. The basic algorithm colors the edges of the graph with $r$ colors and orients each color separately with a maximum out-degree of 1. Then, the combined out-degree of any vertex is at most $r$. We analyze the probe complexity by deriving a Chernoff-style concentration bound on the size of any connected component of any color.

The probe complexity bound is meaningful when $r$ is at least polynomial in $\D$; otherwise, if $r \ge \D$, the problem is trivial. 
For example, if $\D=100$ and $r=10$, or if $\D = \log n$ and $r = \sqrt{\log n}$, then the resulting probe complexity is $n^{1/2+o(1)}$. Even if $\D$ is polynomial in $n$, this algorithm can still be nontrivial (though the $\D$ factor does become important). For example, if $\D = n^{0.2}$ and $r = n^{0.1}$, then the complexity is $n^{0.7+o(1)}$.

Specifically, we show the following theorem.

\begin{theorem} \label{thm:ub-bdd}
Given parameters $r \le \D$, there exists an LCA which $r$-orients an input forest $G$ with $n$ vertices and maximum degree $\D$. For any fixed $\eps$ and with very high probability, this algorithm uses at most $\D n^{1 - \log_\D r + \e}$ probes for every query, for large enough $n$.
\end{theorem}

We show this by showing a concentration bound on the size of connected components in a random subset of a tree:

\begin{prop}\label{prop:tree-chernoff}
Fix a constant $\e$. Suppose the edges of a forest $G$ with $n$ vertices and maximum degree $\D$ are each colored independently with probability $p \ge \D^{-1+\eps}$. Then, with very high probability (in $n$), every connected component of colored edges has size at most $n^{1+\log_\D p + \eps}$.
\end{prop}

Before we prove \cref{prop:tree-chernoff}, we show how we can deduce \cref{thm:ub-bdd} from it:

\begin{proof}[Proof of \cref{thm:ub-bdd} assuming \cref{prop:tree-chernoff}]

First, assume that $r \le \D^{1-\eps/2}$ (we will handle the other case at the end). Then the algorithm is given in \cref{alg:bdd}.

\begin{algorithm} \label{alg:bdd}
    \Input{Forest $G$ with $n$ vertices and maximum degree $\D$, with probe access \\
        Query edge $e=(u, v)$, where $u$ has lower ID than $v$
    }
    
    Using the shared randomness, assign an independent and uniformly randomly chosen color from $\{1, \dots, r\}$ to every pair of vertices in $G$ (so that every edge computes the same coloring).
    
    Let $c$ be the color of $e$.
    
    Using adjacency list probes, perform a depth-first search to find the connected component of $e$ within the edges of $G$ with color $c$.
    
    Orient $e$ toward the minimum-ID vertex in the connected component (which must be a tree). \tcp{This ensures that every vertex has at most 1 edge of color $c$ oriented away from itself.}
    
	\caption{
        LCA for $r$-orienting a bounded-degree forest
    }
\end{algorithm}

Since each vertex has at most 1 edge of color $c$ oriented away from itself, the total out-degree of any vertex is at most $r$, as desired.

Using \cref{prop:tree-chernoff} (with $\eps/2$ instead of $\eps$), with high probability this has size at most $n^{1 - \log_\D r + \e/2}$ (with a union bound over all colors). Then, the depth-first search (using the adjacency lists) takes at most $\D n^{1 - \log_\D r + \e/2}$ probes since it must probe the entire adjacency list of every vertex in this connected component. 

Finally, we handle the case where $\D^{1-\eps/2} \le r \le \D$ (if $r \ge D$ then an arbitrary orientation works). By the previous logic, we can still get a $\D^{1-\eps/2}$-orientation (which is also an $r$-orientation) with probe complexity $\D n^{\eps}$. But we have $\D n^{\eps} \le \D n^{1 - \log_\D r + \e}$, so we are done.
\end{proof}

\subsection{Proof of \cref{prop:tree-chernoff}}
By taking a union bound over connected components, we can assume $G$ to be a tree. Pick an arbitrary fixed vertex $v$, color it, and consider $G$ to be a rooted tree rooted at $v$. We will show that with very high probability the colored connected component at $v$ has size at most $n^{1+\log_\D p + \eps}$. We prove this using a Chernoff bound-style method.
Indeed, if we prove the following claim, then we are done by a union bound over $v$.

\begin{claim}
Fix a constant $\eps$, and let $n$ be sufficiently large. Let $T$ be a rooted tree with $n$ vertices and with maximum degree $\D$. Color each edge of $T$ independently with probability $p > \D^{-1+\eps}$ and let $X_T$ be the number of vertices in the colored connected component of the root of $T$.
Then, we have 
\[\Pr[X_T \ge n^{1+\log_\D p + \eps}] < e^{-n^{\eps/5}}.\]
\end{claim}
\begin{proof}

Let $t = n^{-1-\log_\D p - 4\eps/5}$, and let $f(m) = t n^{2\e/5} m^{1 + \log_\D p + 2n^{-\eps/5}/\log \D}$. Note that the exponent on $m$ is between 0 and 1 for large enough $n$; this means that $f$ is concave and increasing. We will show that $\Ex[e^{X_T t}] \le 1 + f(n)$.

We show more generally that if $T$ has $m \le n$ vertices, then $\Ex[e^{X_T t}] \le 1 + f(m)$. We show this by induction on $m$; this is obviously true for $m=0$ (the vacuous case of an empty tree where $X_T$ is always 0). 

Let $T_1, \dots, T_\D$ be the subtrees of $T$, with sizes $m_1, \dots, m_\D$ (some of them may be empty with $m_i=0$). The sum of the $m_i$ is then $m-1$. Then we can write
\[X_T = 1 + Y_1X_{T_1} + Y_2X_{T_2} + \dots + Y_\D X_{T_\D},\]
where each $Y_i$ a Bernoulli variable with probability $p$ indicating whether the edge from the root of $T$ to the root of $T_i$ is colored. In the case that $T_i$ is the empty tree, we can let $Y_i$ still be a Bernoulli variable with probability $p$. Note that all the random variables in the right hand side above are independent.

We have that $e^{Y_i X_{T_i}t}$ is equal to $e^{X_{T_i}t}$ with probability $p$ (independent of $X_{T_i}$) and equal to 1 otherwise. Thus, we can write
\[
    \Ex[e^{Y_i X_{T_i}t}] = 1 - p + p\Ex[e^{X_{T_i}t}] \le 1+pf(m_i),
\]
since we have $\Ex[e^{X_{T_i}t}] \le 1+ f(m_i)$ by the inductive hypothesis. Therefore,
\begin{align*}
\Ex[e^{X_T t}] 
&= e^t \prod_{i=1}^\D \Ex[e^{Y_i X_{T_i}t}] \\
&\le e^t \prod_{i=1}^\D (1+ pf(m_i)) \\
&\le \exp\p{t + \sum_{i=1}^\D pf(m_i)}.
\end{align*}
The function $f$ is concave and increasing, and the sum of the $m_i$ is $m-1$, so we can bound the above as follows:
\begin{align*}
\Ex[e^{X_T t}]
&\le \exp\p{t + p\D f\paf{m-1}{\D}} \\
&\le \exp \p{t + p\D f\paf{m}{\D}}.
\end{align*}
We first bound the quantity inside the exponential:
\begin{align*}
t + p\D f\paf{m}{\D}
&= t + p \D t n^{2\e/5} \paf{m}{\D}^{1 + \log_\D p + 2n^{-\e/5}/\log \D} \\
&= t + p \D t n^{2\e/5} \cdot \f{m^{1 + \log_\D p + 2n^{-\e/5}/\log \D}}{p\D e^{2n^{-\e/5}}} \\
&= t + t n^{2\e/5} {e^{-2n^{-\e/5}}} m^{1 + \log_\D p + 2n^{-\e/5}/\log \D} \\
&\le \p{1+n^{-2\e/5}e^{2n^{-\e/5}}} t n^{2\e/5} {e^{-2n^{-\e/5}}} m^{1 + \log_\D p + 2n^{-\e/5}/\log \D} \\
&\le \p{1+n^{-\e/5}} t n^{2\e/5} {e^{-2n^{-\e/5}}} m^{1 + \log_\D p + 2n^{-\e/5}/\log \D} \\
&\le \p{1+n^{-\e/5}}^{-1} t n^{2\e/5} m^{1 + \log_\D p + 2n^{-\e/5}/\log \D} \\
&= \p{1+n^{-\e/5}}^{-1} f(m).
\end{align*}

Here, in the first inequality, we have used the fact that $m \ge 1$ and its exponent is positive. Let $x = (1+n^{-\e/5})^{-1} f(m)$; we then have $\Ex[e^{X_T t}] \le e^x$, so it remains to show that $e^x \le 1 + f(m)$. Since $x$ is slightly smaller than $f(m)$, we show this by showing that $x$ is sufficiently small. Indeed, we have
\[x \le f(m) \le f(n) = n^{-2\e/5 + 2n^{-\eps/5}/\log \D} \le n^{-\eps/5}.\]
In particular, this means that $x < 1$, so \[e^x \le 1 + x + x^2 \le 1 + (1 + n^{-\eps/5})x = 1+ f(m).\]
Therefore we have shown that $\Ex[e^{X_T t}] \le 1+f(m)$, completing the induction.

Now, we have $\Ex[e^{X_T t}] \le 1 + f(n)$ for the original tree with $n$ vertices. As above, we have $f(n) \le n^{-\eps/5} < 1$. However, if $X_T \ge n^{1+\log_\D p + \eps}$, then we would have $e^{X_T t} \ge e^{n^{\eps/5}}$. Thus, by Markov's inequality, the probability that $X_T \ge n^{1+\log_\D p + \eps}$ is less than $e^{-n^{\eps/5}}$. This completes the proof of the claim, and thus the proof of \cref{prop:tree-chernoff}.
\end{proof}

\subsection{Coloring bounded-degree forests}
The concentration bound in \cref{prop:tree-chernoff} also has implications for vertex-coloring in bounded-degree forests. In particular, we show that a randomized LCA can 4-color any graph with bounded (constant) degree with sublinear probe complexity. This is in contrast to the lower bound of \cite{brandt2021randomized}, which shows that it is impossible to do so with a sublinear deterministic VOLUME algorithm (VOLUME is a model very similar to, but slightly weaker than, the LCA model).

\begin{theorem} \label{thm:color}
Let $r, \D$ be fixed constants with $2 \le r \le \log_2 \D$. There exists an LCA which, given a forest $G$ with $n$ vertices and maximum degree $\D$, $2^r$-colors its vertices. For any fixed $\eps$ and with very high probability, this algorithm uses at most $n^{1 - \log_\D r + \e}$ probes for every query, for large enough $n$.
\end{theorem}
\begin{proof}
In our algorithm for this proof, we randomly label each edge with one of $r$ possible labels, and 2-color each connected component of each label. Then we take the product coloring of all these colorings (one coloring for each label). This algorithm is described precisely in \cref{alg:color}.

\begin{algorithm} \label{alg:color}
    \Input{Forest $G$ with $n$ vertices and maximum degree $\D$, with probe access \\
        Query vertex $v$
    }
    
    Using the shared randomness, assign an independent and uniformly randomly chosen label from $\{1, \dots, r\}$ to every pair of vertices in $G$ (so that every edge computes the same coloring).
    
    \ForEach{label $i$ in $\{1, \dots, r\}$} {
        Perform a depth-first search to find the connected component of $u$ within the edges of $G$ with label $c$.
        
        Compute the 2-coloring of this connected component (which must be a tree) such that the minimum-ID vertex has color 1.
        
        Let $c_i \in \{1, 2\}$ be the color assigned to $v$ under this coloring.
    }
    
    Output the color $(c_1, \dots, c_r)$. \tcp{Note that there are $2^r$ possible colors.}
    
    \caption{
        LCA for $2^r$-vertex-coloring a bounded-degree forest
    }
\end{algorithm}

This algorithm produces a proper $2^r$-coloring since for any $u$ and $v$, if $(u, v)$ has label $i$, then the $i$-th indices of their colors must differ.
If $\e$ is sufficiently small (as we may assume), then we have that $1/r \ge 1/\log_2 \D \ge \D^{-1+\eps/2}$, so we may apply \cref{prop:tree-chernoff}, with $\e/2$ instead of $\e$. Then, with very high probability, every connected component of every label has size at most $n^{1 - \log_\D r + \e/2}$; for this claim, we are taking a union bound over the colors. Then, the depth-first searches in total use at most $r\D n^{1 - \log_\D r + \e/2} \le n^{1 - \log_\D r + \e}$ probes, which completes the proof.

\end{proof}


\section*{Acknowledgements}
We are grateful to anonymous reviewers for many helpful comments on the earlier draft of this submission.

\bibliographystyle{alpha}
\bibliography{ref}


\end{document}